\documentclass[11pt]{article}

\usepackage{graphicx}
\usepackage[margin=1in]{geometry}
\usepackage{microtype}
\usepackage{amssymb}
\usepackage{amsthm}
\usepackage{amsmath}
\usepackage{xspace}
\usepackage{epsfig}
\usepackage{ifpdf}
\usepackage{url,hyperref}
\usepackage{latexsym}
\usepackage{xspace}
\usepackage{color}
\usepackage{makeidx}
\usepackage{stackrel}
\usepackage{amscd}
\usepackage[all]{xy}
\usepackage{lineno}
\usepackage{algorithm}
\usepackage{algpseudocode}

\usepackage{framed}
\usepackage{pb-diagram}

\usepackage[dvipsnames]{xcolor}
\usepackage[labelformat=parens]{subfig}
\graphicspath{ {figs/} }

\long\def\remove#1{}

\theoremstyle{plain}
\newtheorem{theorem}{Theorem}
\newtheorem{proposition}{Proposition}
\theoremstyle{remark}
\newtheorem{remark}{Remark}
\theoremstyle{definition}
\newtheorem{definition}{Definition}
\newtheorem{problem}{Problem}
\newtheorem{algr}[algorithm]{Algorithm}

\makeatletter
\newcommand{\algmargin}{\the\ALG@thistlm}
\makeatother
\newlength{\whilewidth}
\algdef{SE}[parWHILE]{parWhile}{EndparWhile}[1]
  {\parbox[t]{\dimexpr\linewidth-\algmargin}{%
     \hangindent\whilewidth\strut\algorithmicwhile\ #1\ \algorithmicdo\strut}}{\algorithmicend\ \algorithmicwhile}%
\algdef{SE}[parIF]{parIf}{EndparIf}[1]
  {\parbox[t]{\dimexpr\linewidth-\algmargin}{%
     \hangindent\whilewidth\strut\algorithmicif\ #1\ \algorithmicthen\strut}}{\algorithmicend\ \algorithmicif}%
\algnewcommand{\parState}[1]{\State%
  \parbox[t]{\dimexpr\linewidth-\algmargin}{\strut #1\strut}}

\long\def\remove#1{}

\definecolor{darkblue}{rgb}{0.0, 0.0, 0.8}
\definecolor{darkred}{rgb}{0.8, 0.0, 0.0}
\definecolor{darkgreen}{rgb}{0.0, 0.8, 0.0}

\newcommand {\mm}[1] {\ifmmode{#1}\else{\mbox{\(#1\)}}\fi}

\newcommand{\Z}                 {\mathrm {\mathbb{Z}}}

\newcommand{\Hm}{\mathrm{\sf H}}
\newcommand{\Real}{\mathbb{R}}

\newcommand{\Pm}{\mathcal{P}}
\newcommand{\Intmod}{\mathcal{I}}
\newcommand{\Bas}{\mathcal{B}}
\newcommand{\Vect}{\mathrm{\sf Vect}}

\newcommand{\Simp}{\mathrm{\sf Simp}}

\newcommand{\Fcal}{\mathcal{F}}
\newcommand{\Ical}{\mathcal{I}}
\newcommand{\Qcal}{\mathcal{Q}}

\usepackage{mathtools}
\newcommand\SetSymbol[1][]{\nonscript\:#1\vert\allowbreak\nonscript\:\mathopen{}}
\providecommand\given{} 
\DeclarePairedDelimiterX\Set[1]\{\}{\renewcommand\given{\SetSymbol[\delimsize]}#1}

\DeclareMathOperator{\persdgm}{D}

\let\emptyset\varnothing

\renewcommand{\And}{\textbf{and}\xspace}
\newcommand{\Or}{\textbf{or}\xspace}

\begin{document}

\title{Persistent $1$-Cycles: Definition, Computation, and Its Application\thanks{Supported by NSF grants CCF-1740761 and CCF-1839252.}}

\author{Tamal K. Dey\thanks{Department of Computer Science and Engineering, The Ohio State University. \texttt{dey.8@osu.edu}}
\and Tao Hou\thanks{Department of Computer Science and Engineering, The Ohio State University. \texttt{hou.332@osu.edu}}
\and Sayan Mandal\thanks{Department of Computer Science and Engineering, The Ohio State University. \texttt{mandal.25@osu.edu}}
}

\date{}

\maketitle              

\begin{abstract}
Persistence diagrams, which summarize the birth and death of homological features extracted from data, are employed as stable signatures for applications in image analysis and other areas. 
Besides simply considering the multiset of intervals included in a persistence diagram, some applications need to find representative cycles for the intervals.
In this paper, we address the problem of computing these representative cycles, termed as {\em persistent $1$-cycles}, for $\Hm_1$-persistent homology with $\Z_2$ coefficients. 
The definition of persistent cycles is based on the interval module decomposition of persistence modules, which reveals the structure of persistent homology.
After showing that the computation of the optimal persistent $1$-cycles is NP-hard, 
we propose an alternative set of {\it meaningful} persistent $1$-cycles that can be computed with an efficient polynomial time algorithm. 
We also inspect the stability issues of the optimal persistent $1$-cycles and the persistent $1$-cycles computed by our algorithm with the observation that the perturbations of both cannot be properly bounded.
We design a software which applies our algorithm to various datasets. Experiments on 3D point clouds, mineral structures, and images show the effectiveness of our algorithm in practice. 

\end{abstract}

\newpage
\setcounter{page}{1}

\section{Introduction}
\label{sec:intro}
Persistent homology \cite{edelsbrunner2000topological} is an important invention leading to Topological Data Analysis, where the associated persistence diagrams serve as stable signatures for various datasets \cite{cohen2005stability} including the ones in image analysis {\cite{Carlsson:2008,Dey2017}}. 
Persistent homology has its theoretical foundations rooted in quiver theory {\cite{derksen2005quiver}}, in which case any persistence module indexed by a finite subcategory of $\Real$ can be decomposed into a direct sum of interval modules and the set of intervals of the interval modules, which constitute the persistence diagram, is unique for a persistence module \cite{chazal2016structure}. 
\begin{figure}
\centering
  \subfloat[]{\includegraphics[width=0.15\linewidth]{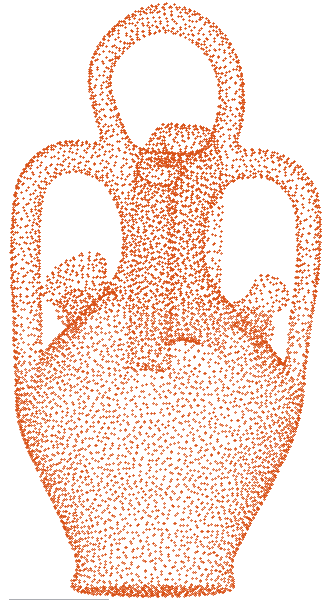}\label{fig:intro1}}
  \hspace{0.1cm}
  \subfloat[]{\includegraphics[width=0.2\linewidth]{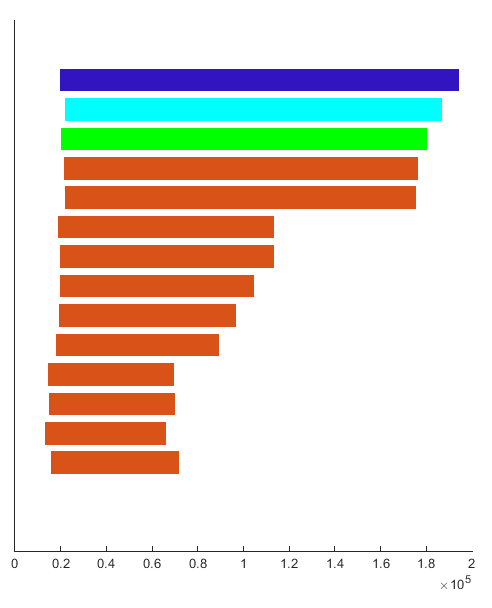}\label{fig:intro2}}
  \vspace{0.1cm}
  \subfloat[]{\includegraphics[width=0.15\linewidth]{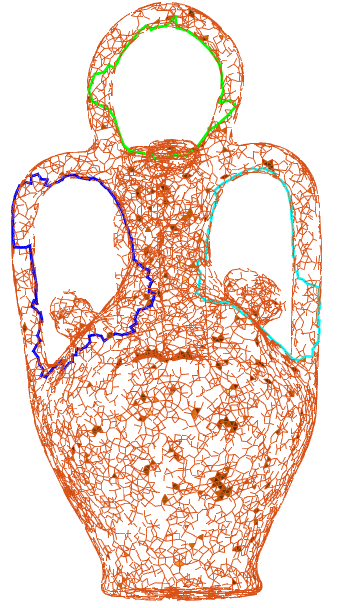}\label{fig:intro3}}
  \vspace{0.1cm}
  \subfloat[]{\includegraphics[width=0.2\linewidth]{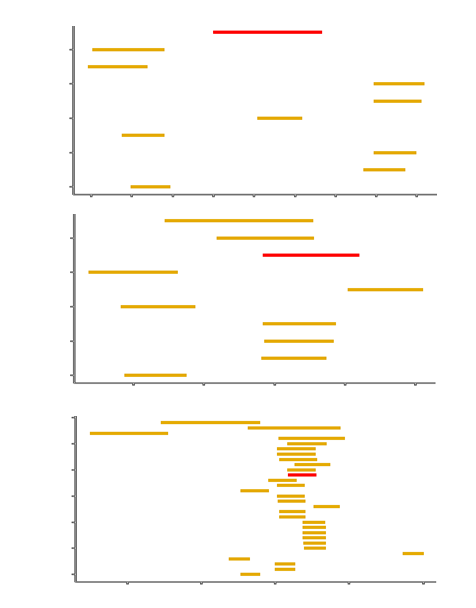}\label{fig:intro4}}
  \vspace{0.1cm}
  \subfloat[]{\includegraphics[width=0.12\linewidth]{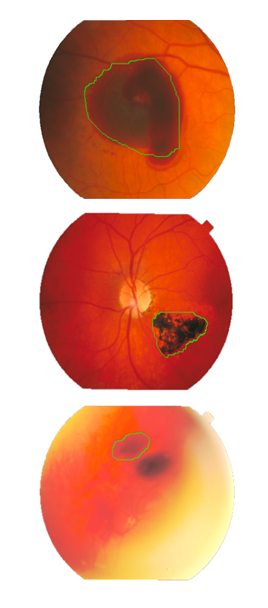}\label{fig:intro5}}

   \caption{
  (a) Point cloud of Botijo model. 
  (b,c) Barcode and persistent $1$-cycles for Botijo, where the 3 longest bars (dark blue, light blue, and green) have their corresponding persistent $1$-cycles drawn with the same colors.
  (d,e) Barcode and persistent $1$-cycles for the retinal image, with each green cycle corresponding to a red bar. }\label{fig:intro}
\end{figure}


Besides simply incorporating the persistence diagrams, some applications bring about the need of finding representative cycles for persistent homology {\cite{emmett2016multiscale,wu2017optimal}}. 
The computation of representative cycles for homology groups with $\Z_2$ coefficients has been extensively studied over the decades. 
While a polynomial time algorithm computing an optimal basis for first homology group $\Hm_1$ \cite{dey2010approximating} has been proposed, 
finding an optimal basis for dimension greater than one and localizing a homology class of any dimension are proved NP-hard \cite{chen2011hardness}.
There are a few works addressing the problem of finding representatives for persistent homology,
some  of which compute an optimal cycle at the birth index of an interval but do not consider what actually die at the death index~\cite{emmett2016multiscale,escolar2016optimal}.
Obayashi \cite{obayashi2017volume} formalizes the computation of optimal representatives for a finite interval as an integer programming problem. He advocates solving it with linear programs though the correctness is not necessarily guaranteed. 
Wu et al. \cite{wu2017optimal} proposed an algorithm for computing an optimal representative for a finite interval with a worst-case complexity exponential to the cardinality of the persistence diagram.

In this paper, we study the problem of computing representative cycles for persistent first homology group ($\Hm_1$-persistent homology) with $\Z_2$ coefficients. 
We term theses cycles as {\em persistent $1$-cycles} and show that the computation of the optimal cycles is NP-hard.
Then, we propose an alternative set of {\em meaningful} persistent $1$-cycles with an efficient polynomial time algorithm. 
Specifically, as interval module decomposition reveals the structure of persistence modules, we define persistent cycles which fit into this structure directly. 
Although similar definitions for finite intervals have already been proposed~\cite{obayashi2017volume,wu2017optimal}, to our knowledge, explicit explanation of how the representative cycles are related to persistent homology has not been addressed.
Furthermore, we inspect the stability of the minimal persistent $1$-cycles and persistent $1$-cycles computed by our algorithm. 
The perturbations of both classes of cycles turn out to be unstable. So, in this regard, our polynomial time algorithm is not any worse than an optimal cycle generating algorithm though is much more efficient in terms of the time complexity.

We use a software based on our algorithm to generate tight persistent $1$-cycles on 3D point clouds and 2D images as shown in Figure~\ref{fig:intro}. We experiment with various datasets commonly used in geometric modeling, computer vision and material science, details of which are given in Section~\ref{sec:experiment}. The \textit{software, named} \texttt{PersLoop}\textit{, along with an introductory video} and other supplementary materials are available at the project website 
\url{http://web.cse.ohio-state.edu/~dey.8/PersLoop}.

\section{Background}
\label{sec:background}

In this paper, we adopt the categorical definition of persistence module \cite{bubenik2014categorification}. 
A category $C$ consists of objects and morphisms from an object to another object. 
A functor $F:C\rightarrow B$ from $C$ to another category $B$ is a mapping 
such that any object $c$ of $C$ is mapped to an object $F(c)$ of $B$ 
and any morphism $f:c\rightarrow c'$ of $C$ is mapped to a morphism $F[f]:F(c)\rightarrow F(c')$ of $B$. 
We recommend \cite{awodey2010category} for the exact definitions of categories and functors. 
The definition of persistence module relies on some common categories: 
The category $\Z^+$ (the category $\Set{1,\ldots,n}$ alike) consists of objects from $\Z^+$ 
and a unique morphism from $i$ to $j$ if $i\leq j$. 
We also denote the morphism from $i$ to $j$ as $i\leq j$. 
The category $\Simp$ consists of objects which are all the simplicial complexes and morphisms which are simplicial maps.
The category $\Vect$ consists of objects which are all the vector spaces over $\Z_2$ and morphisms which are linear maps.
A persistence module $\Pm$ is then defined as a functor $\Pm:\Z^+\rightarrow\Vect$\footnote{
Sometimes we also call a functor $\Pm:\Set{1,\ldots,n}\rightarrow\Vect$ as a persistence module.}.

A persistence module is usually induced by a filtration $\Fcal=\Fcal(K)$ of a simplicial complex $K$, 
where the filtration $\Fcal:\emptyset=K_0\subseteq K_1\subseteq\ldots\subseteq K_m=K$ is a filtered sequence of subcomplexes of $K$ such that $K_{i+1}$ and $K_{i}$ differ by one simplex $\sigma_{i+1}$.
We can also interpret a filtration $\Fcal$ as a functor $\Fcal:\Z^+\rightarrow\Simp$, 
where $\Fcal(i)=K_i$ for $i\leq m$, $\Fcal(i)=K$ for $i>m$, 
and a morphism $\Fcal[i\leq j]:\Fcal(i)\rightarrow\Fcal(j)$ is the inclusion.
Denoting $\Hm_q:\Simp\rightarrow\Vect$ as the {$q_{\text{th}}$} homology functor with $\Z_2$ coefficients, 
the {$\Hm_q$-}persistence module {$\Pm_q^{\Fcal}$} of $\Fcal$ is obtained by composing the two functors $\Hm_q$ and $\Fcal$, that is, $\Pm_q^\Fcal=\Hm_q\Fcal$. 
Specifically, $\Pm_q^\Fcal(i)=\Hm_q(K_i)$ for $i\leq m$, $\Pm_q^\Fcal(i)=\Hm_q(K)$ for $i>m$, 
and the morphism $\Pm_q^\Fcal[i\leq j]:\Hm_q(K_i)\rightarrow\Hm_q(K_j)$\footnote{$K_j=K$ when $j>m$.}
is the linear map induced by the inclusion.

A special class of persistence modules is the interval modules. 
Given an interval $[b,d)\subset\Z^+$, an interval module $\Ical^{[b,d)}$ is defined as: 
$\Ical^{[b,d)}(i)=\Z_2$ for $i\in[b,d)$ and $\Ical^{[b,d)}(i)=0$ otherwise; 
$\Ical^{[b,d)}[i\leq j]$ is the identity map for $i,j\in[b,d)$ and $\Ical^{[b,d)}[i\leq j]$ is the zero map otherwise. 
{By quiver theory, a $\Hm_q$-persistence module obtained from a finite complex $K$ has a unique decomposition $\Pm_q^\Fcal=\bigoplus_{j\in J}\Intmod^{[b_j,d_j)}$ in terms of interval modules, 
where $J\subset \Z$ is a finite index set \cite{chazal2016structure}.} 
Let $\persdgm(\Pm_q^\Fcal)=\Set{[b_j,d_j)\given j\in J}$ denote the set of intervals of the interval modules which $\Pm_q^\Fcal$ decomposes into. 
Observe that $\persdgm(\Pm_q^\Fcal)$ is also called the {\it barcode} or {\it persistence diagram} in the literature \cite{edelsbrunner2010computational}. 
Sometimes we will abuse the notation slightly to write $\persdgm_q(\Fcal)$, where the argument is the filtration instead of the module $\Pm_q^\Fcal$ it generates.


\section{Persistent basis and cycles}

\begin{definition}[Persistent Basis]
\label{dfn:pers_bas}
An indexed set of $q$-cycles $\Set{c_j\given j\in J}$ is called a persistent $q$-basis for a filtration $\Fcal$ if $\Pm_q^{\Fcal}=\bigoplus_{j\in J}\Intmod^{[b_j,d_j)}$ and for each $j\in J$ and $b_j\leq k<d_j$, $\Intmod^{[b_j,d_j)}(k)=\Set{0,[c_j]}$.
\end{definition}


\begin{definition}[Persistent Cycle]
\label{dfn:pers_cyc}
For an interval $[b,d)\in \persdgm(\Pm_q^{\Fcal})$, a $q$-cycle $c$ is called a persistent $q$-cycle for the interval, 
if one of the following holds:
\begin{itemize}
\item $d\neq +\infty$, $c$ is a cycle in $K_{b}$ containing $\sigma_{b}$, and $c$ is not a boundary in $K_{d-1}$ but becomes a boundary in $K_{d}$;
\item $d= +\infty$ and $c$ is a cycle in $K_{b}$ containing $\sigma_{b}$.
\end{itemize}
\end{definition}

\begin{remark}
Note that the definition of persistent cycles for finite intervals is identical to that of \cite{obayashi2017volume,wu2017optimal}.
\end{remark}

The following theorem characterizes each cycle in a persistent basis:

\begin{theorem}
\label{thm:pers_bas_charac}
An indexed set of $q$-cycles $\Set{c_j\given j\in J}$ is a persistent $q$-basis for a filtration $\Fcal$ if and only if $\Pm_q^{\Fcal}=\bigoplus_{j\in J}\Intmod^{[b_j,d_j)}$ and $c_j$ is a persistent $q$-cycle for every interval $[b_j,d_j)\in \persdgm(\Pm_q^{\Fcal})$. 
\end{theorem}
\begin{proof}
Suppose $\Set{c_j\given j\in J}$ is an indexed set of $q$-cycles satisfying the above conditions. For each $j\in J$, we construct an interval module $\Intmod_j$, such that $\Intmod_j(i)=\{0,[c_j]\}$ for $b_j\leq i< d_j$ and $\Intmod_j(i)={0}$ otherwise. We claim that $\Pm_q^\Fcal=\bigoplus_{j\in J}\Intmod_j$. 
We first prove that $\Pm_q^\Fcal(i)=\bigoplus_{j\in J}\Intmod_j(i)$ for each $i\in\Z^+$, by proving that $\Set{[c_j]\given j\in J, i\in [b_j,d_j)}$ forms a basis of $\Pm_q^\Fcal(i)$. Using mathematical induction, since $\sigma_1$ is a vertex, this is trivially true. Suppose for $i-1$ this is true. 
If $\sigma_i$ is neither positive nor negative, i.e., $\Hm_q(K_{i-1}) \approx \Hm_q(K_{i})$ by the isomorphism induced from the inclusion, this is also trivially true for $i$. 
If $\sigma_i$ is positive, suppose the corresponding interval of $\sigma_i$ is $[b_{j'},d_{j'})$ (note that $b_{j'}=i$ and $d_{j'}$ could possibly be $+\infty$). 
Since $\Set{[c_j]\given j\in J, i-1\in [b_j,d_j)}$ are still independent in $\Pm_q^\Fcal(i)$ and $[c_{j'}]$ is not in the span of them, then $\Set{[c_j]\given j\in J, i-1\in [b_j,d_j)}\cup {[c_{j'}]}=\Set{[c_j]\given j\in J, i\in [b_j,d_j)}$ are independent in $\Pm_q^\Fcal(i)$. 
Since the cardinality of $\Set{[c_j]\given j\in J, i\in [b_j,d_j)}$ equals the dimension of $\Pm_q^\Fcal(i)$, it must form a basis of $\Pm_q^\Fcal(i)$. 
If $\sigma_i$ is negative, then there must be a $[c_{j'}]$ for a $j'\in J$ such that $d_{j'}=i$. 
For any $[c]\in \Pm_q^\Fcal(i)=\Hm_q(K_i)$, $[c]=\sum_{j\in J'}[c_j]$, where $J'\subseteq \Set{j\in J\given i-1\in [b_j,d_j)}$. 
If $j'\in J'$, then $[c]=\sum_{j\in J'-\{j'\}}[c_j]$, because $[c_{j'}]=0$ in $\Hm_q(K_i)$. 
Then $\Set{[c_j]\given j\in J, i-1\in [b_j,d_j)}-\{c_{j'}\}=\Set{[c_j]\given j\in J, i\in [b_j,d_j)}$ spans $\Hm_q(K_i)$. This means that it also forms a basis of $\Hm_q(K_i)$. It is then obvious that the direct sums of the maps of the interval modules are actually the maps of $\Pm_q^\Fcal$, so $\Set{c_j\given j\in J}$ is a persistent $q$-basis for $\Fcal$.

Suppose $\Set{c_j\given j\in J}$ is a persistent $q$-basis for $\Fcal$. For each $j\in J$, $c_j$ must not be in $K_{b_j-1}$, because otherwise $[c_j]$ would be in the image of $\Pm_q^\Fcal[b_j-1\leq b_j]$. It is obvious that $c_j$ must contain $\sigma_j$. Note that for each $j\in J$ and each $i\in [b_j,d_j)$, $\Pm_q^\Fcal[i\leq i+1]([c_j])=\Intmod^{[b_j,d_j)}[i\leq i+1]([c_j])$.
Then for each $j\in J$ such that $d_j\neq +\infty$, $[c_j]\neq 0$ in $K_{d_j-1}$ and $[c_j]=0$ in $K_{d_j}$.
\end{proof}

With Definition \ref{dfn:pers_cyc} and Theorem \ref{thm:pers_bas_charac}, it is true that for a persistent $q$-cycle $c$ of an interval $[b,d)\in\persdgm_q(\Fcal)$, we can always form an interval module decomposition of $\Pm_q^\Fcal$, where $c$ is a representative cycle for the interval module of $[b,d)$. 

\section{Minimal persistent $q$-basis and their limitations}
We have already defined persistent basis, the optimal versions of which are of particular interest because they capture more geometric information of the space. The cycles for an optimal (minimal) persistent basis have already been defined and studied in \cite{escolar2016optimal,obayashi2017volume}. In particular, the author of \cite{obayashi2017volume} proposed an integer program to compute these cycles. Although these integer programs can be solved exactly by linear programs for certain cases \cite{DHK11}, the integer program is NP-hard in general. This of course does not settle the question of whether the problem of computing minimal persistent $1$-cycles is NP-hard or not. We prove that it is indeed NP-hard and thus has no hope of admitting a polynomial time algorithm unless $\text{P}=\text{NP}$.

Consider a simplicial complex $K$ with each edge being assigned a non-negative weight. We refer to such $K$ as a weighted complex. For a $1$-cycle $c$ in $K$, define its weight to be the sum of all weights of its edges.
\begin{definition}[Minimal Persistent $1$-Cycle and $1$-Basis]
Given a filtration {$\Fcal$} on a weighted complex $K$,
a minimal persistent $1$-cycle for an interval of {$\persdgm_1(\Fcal)$} is defined to be a persistent $1$-cycle for the interval with the minimal weight. An indexed set of $1$-cycles $\Set{c_j\given j\in J}$ is a minimal persistent $1$-basis for $\Fcal$ if for every $[b_j,d_j)\in \persdgm_1(\Fcal)$, $c_j$ is a minimal persistent $1$-cycle for $[b_j,d_j)$.
\end{definition}

We prove that the following special version of the problem of finding a minimal persistent $1$-cycle is NP-hard. This special version reduces to the general version straightforwardly in polynomial time by assigning every edge a weight of $1$.

\begin{problem}[LST-PERS-CYC]
Given a filtration $\Fcal: \emptyset=K_0\subseteq K_1\subseteq\ldots\subseteq K_m=K$, and a finite interval $[b,d)\in\persdgm_1(\Fcal)$, find a $1$-cycle with the least number of edges which is born in $K_{b}$ and becomes a boundary in $K_{d}$.
\end{problem}

\begin{theorem} 
\label{thm:nphard}
The problem {LST-PERS-CYC} is NP-hard
\end{theorem}

The proof of Theorem \ref{thm:nphard} is given in Appendix \ref{app:nphard-pf}.

\subsection{Instability of minimal persistent $1$-cycles}
\label{sec:instab_optim}
In this section, we inspect the stability issues of the minimal persistent $1$-cycles. 
Note that there may be multiple minimal persistent $1$-cycles for an interval and an algorithm may choose anyone of them. 
This means that the cycles cannot be stable under those measures that take into account the entire geometry of the cycles (e.g., Hausdorff distance).
In an attempt to sidestep this problem, we take a `weaker' measure of the cycles which is still meaningful, namely their lengths.  
We show that even under such a measure, minimal persistent $1$-cycles are unstable.
Specifically, we consider the lower star filtration \cite{edelsbrunner2010computational} of a vertex sequence, and inspect the perturbation of the lengths of persistent $1$-cycles under the perturbation of the sequence.
Since each interval $I$ in the $\Hm_1$-persistence diagram of a lower star filtration can be derived from an interval $I'$ in the $\Hm_1$-persistence diagram of a corresponding insertion filtration\footnote{
The \textit{insertion filtration} is actually the filtration defined in Section \ref{sec:background}.}, 
we can associate a persistent $1$-cycle for $I'$ to $I$. 
{The readers can verify that this assignment gives representatives for the decomposed interval modules of the $\Hm_1$-persistence module induced by the lower star filtration.}

\begin{figure}
    \centering
    \subfloat
      []{\includegraphics[width=0.3\linewidth]
      {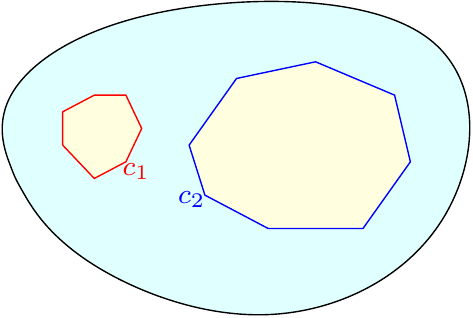}\label{fig:instab_optim}} \hspace{3em}
    \subfloat
      []{\includegraphics[width=0.3\linewidth]
      {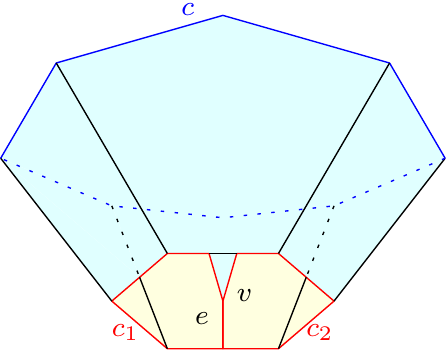}\label{fig:instab_local}}
    \caption{(a) A sphere with two holes shows the instability of the minimal persistent $1$-cycles. (b) The simplicial complex shows the instability of the cycles of Algorithm \ref{alg:all_int}. $c_1$ and $c_2$ are the two red simple cycles; $e$ is the edge adjacent to $v$.}
\end{figure}

Figure \ref{fig:instab_optim} presents an example for which the perturbation of the minimal persistent $1$-cycles cannot be properly bounded. 
The object in Figure \ref{fig:instab_optim} is a sphere with two holes (i.e., $c_1$ and $c_2$). 
We can assume that the object is nicely triangulated so that it becomes a simplicial complex. 
Let $v_1$ and $v_2$ be vertices from $c_1$ and $c_2$. 
We can construct a filtration\footnote{
Note that we are constructing an insertion filtration for a lower star filtration.}
by first forming the two cycles $c_1$ and $c_2$, 
with $v_1$ and $v_2$ being the last two vertices added, 
then adding the other parts of the simplicial complex. 
We then add a cone around $c_1$ to the filtration. 
We can first assume $v_1$ is added before $v_2$, and the indices of $v_1$ and the apex vertex of the cone in the sequence are $b$ and $d$.
Then the minimal persistent $1$-cycle for the interval $[b,d)$ is $c_1$.
If we switch $v_1$ and $v_2$, the minimal cycle for the interval $[b,d)$ becomes $c_2$. The difference of $c_1$ and $c_2$ can be made arbitrary under a single switch, which is the smallest possible perturbation of lower star filtration.

\section{Computing meaningful persistent $1$-cycles in polynomial time}

Because the minimal persistent $1$-cycles are not stable and their computation is NP-hard, we propose an alternative set of meaningful persistent $1$-cycles which can be computed efficiently in polynomial time.
We first present a general framework. Although the computed persistent $1$-cycles have no guaranteed properties, the framework lays the foundation for the algorithm computing meaningful persistent $1$-cycles that we propose later.

\begin{algr}
\label{alg:diag}
Given a simplicial complex $K$, 
a filtration $\Fcal:\emptyset=K_0\subseteq K_1\subseteq\ldots\subseteq K_m=K$, and $\persdgm_1(\Fcal)$, 
this algorithm finds a persistent $1$-basis for $\Fcal$.
The algorithm maintains a basis $\Bas_i$ for $\Hm_1(K_i)$ for every $i\in[0,m]$. 
Initially, let $\Bas_0=\varnothing$, then do the following for $i=1,\dots,m$:
\begin{itemize}
\item If $\sigma_i$ is positive, find a $1$-cycle $c_i$ containing $\sigma_i$ in $K_i$ and let $\Bas_i=\Bas_{i-1}\cup \{c_i\}$.
\item If $\sigma_i$ is negative, find a set $\Set{c_g\given g\in G}\subseteq \Bas_{i-1}$ so that $\sum_{g\in G}[c_g]=0$. This can be done in $O(\beta_i=|\Bas_i|)$ time by the annotation algorithm in~\cite{dey2014computing}. Maintaining the annotations will take $O(n^{\omega})$ time altogether where $K$ has $n$ simplices and $\omega$ is the matrix multiplication exponent.
Let $g^*$ be the greatest index in $G$, then $[g^*,i)$ is an interval of $\persdgm_1(\Fcal)$.
Assign $\sum_{g\in G}c_g$ to this interval as a persistent $1$-cycle and let $\Bas_i=\Bas_{i-1}\smallsetminus c_{g^*}$.
\item Otherwise, let $\Bas_i=\Bas_{i-1}$.
\end{itemize}
At the end, for each cycle $c_j\in \Bas_m$, assign $c_j$ as a persistent $1$-cycle to the interval $[j,+\infty)$.
\end{algr}

To prove the correctness of the algorithm, we need the following fact:

\begin{proposition}
\label{prop:decomp_iff} 
For a persistence module $\Pm:\{1,\ldots,n\}\rightarrow \Vect$ and a finite set of persistence modules
$\Set{\Qcal_j:\{1,\ldots,n\}\rightarrow \Vect\given j\in J}$, $\Pm=\bigoplus_{j\in J}\Qcal_j$ if and only if
$\Pm(i)=\bigoplus_{j\in J}\Qcal_j(i)$ for each $1\leq i\leq n$ and
$\Pm[i\leq i+1]=\bigoplus_{j\in J}\Qcal_j[i\leq i+1]$ for each $1\leq i< n$.
\end{proposition}

\begin{proof}[Proof of Correctness of Algorithm \ref{alg:diag}]
Denoting all the intervals $[g^*,i)$ found by the algorithm as $D$, we want to inductively prove that for all $i=1,\ldots,m$, the persistence module $\Pm^i$, which is the restriction of $\Pm_1^\Fcal$ to $\Set{1,...,i}$, satisfies:
\begin{equation}
\Pm^i=\bigoplus_{\substack{[b_j,d_j)\in D, d_j\leq i}}\Ical^{[b_j,d_j)} \;\oplus\; \bigoplus_{c_j\in\Bas_i}\Ical^{[j,i]}
\end{equation}
where the representative of $\Ical^{[b_j,d_j)}$ is the persistent $1$-cycle computed by the algorithm and the representative of $\Ical^{[j,i]}$ is $c_j$.
When $i=1$, $\Pm^1$ is trivial and the equation is certainly true. Suppose for $\Pm^i$, the equation is satisfied. If $\sigma_{i+1}$ is neither positive nor negative, or positive, then it is not hard to verify that the equation is still valid for $\Pm^{i+1}$ by Proposition \ref{prop:decomp_iff}. If $\sigma_{i+1}$ is negative, then we can let the persistent $1$-cycle computed by the algorithm for $\sigma_{i+1}$ be $\sum_{g\in G}c_g$ and $g^*$ be the greatest index in $G$. Since $\sum_{g\in G}c_g$ is also created by $\sigma_{g^*}$, we can let the representative of the interval module $\Ical^{[g^*,i]}$ for $\Pm^i$ be $\sum_{g\in G}c_g$. It is not hard then to verify that the equation is still satisfied for $\Pm^{i+1}$ by Proposition \ref{prop:decomp_iff}.
\end{proof}

 Based on Algorithm \ref{alg:diag},
 we present another algorithm
 which produces meaningful persistent $1$-cycles.

\begin{algr}
\label{alg:all_int}
In Algorithm \ref{alg:diag}, when $\sigma_i$ is positive, let $c_i$ be the shortest cycle containing $\sigma_i$ in $K_i$. The cycle $c_i$ can be constructed by adding $\sigma_i$ to the shortest path between vertices of $\sigma_i$ in $K_{i-1}$, which can be computed by Dijkstra's algorithm applied to the $1$-skeleton of $K_{i-1}$.
\end{algr}

Note that in Algorithm \ref{alg:all_int}, a persistent $1$-cycle for a finite interval is a sum of shortest cycles born at different indices. 
Since a shortest cycle is usually a good representative of its class, the sum of shortest cycles ought to be a good choice of representative for an interval. 
In some cases, when $\sigma_i$ is negative, the sum $\sum_{g\in G}c_g$ contains only one component. The persistent $1$-cycles computed by Algorithm \ref{alg:all_int} for such intervals are guaranteed to be optimal as shown below.

\begin{proposition}
\label{prop:alg2_optim}
In Algorithm \ref{alg:all_int}, when $\sigma_i$ is negative, if $|G|=1$, then $\sum_{g\in G}c_g$ is a minimal persistent $1$-cycle for the interval ending with $i$.
\end{proposition}

{In Section \ref{sec:experiment} where we present the experimental results, we can see that, scenarios depicted by Proposition \ref{prop:alg2_optim} occur quite frequently. Especially, for the {larvae and nerve datasets}, nearly all computed persistent $1$-cycles contain only one component and hence are minimal.}

A practical problem with Algorithm \ref{alg:all_int} is that unnecessary computational resource is spent for computing tiny intervals regarded as noise, especially when the user cares about significantly large intervals only. We present a more efficient algorithm for such cases.

\begin{proposition}
\label{prop:ignore_shor_int}
In Algorithm \ref{alg:diag} and \ref{alg:all_int}, when $\sigma_i$ is negative, for any $g\in G$, one has $b_g\leq g^*$ and $d_g\geq i$.
\end{proposition}
\begin{proof}
Note that $\sigma_{b_g}$ must be unpaired before $\sigma_{i}$ is added, this implies that $d_g\geq i$. Since $g^*$ is the greatest index in $G$, $b_g=g\leq g^*$.
\end{proof}

Proposition \ref{prop:ignore_shor_int} leads to Algorithm \ref{alg:single_int} in which we only compute the shortest cycles at the birth indices whose corresponding intervals contain the input interval $[b,d)$. 
Since most of the time the user provided interval $[b,d)$ is a long interval, the intervals containing it constitute a small subset of all the intervals. This makes Algorithm \ref{alg:single_int} run much faster than Algorithm \ref{alg:all_int} in practice.

\begin{algorithm}
\caption{}
\hspace{4pt}\textbf{Input:} The input of Algorithm \ref{alg:all_int} plus an interval $[b,d)$

\hspace{4pt}\textbf{Output:} A persistent $1$-cycle for $[b,d)$ output by Algorithm \ref{alg:all_int}.
\begin{algorithmic}[1]

\State $G'\leftarrow\emptyset$
\For{$i\leftarrow 1,\dots,b$}
  \parIf {$\sigma_i$ is positive \And ($\sigma_i$ is paired with a $\sigma_j$ s.t $j\geq d$ \\ \Or $\sigma_i$ never gets paired)}
    \State $c_i\leftarrow$ the shortest cycle containing $\sigma_i$ in $K_i$
    \State $G'\leftarrow G'\cup \Set{i}$
  \EndparIf
\EndFor
\State find a $G\subseteq G'$ s.t. $\sum_{g\in G}[c_g]=0$ in $K_d$
\State output $\sum_{g\in G}c_g$ as the persistent $1$-cycle for $[b,d)$

\end{algorithmic}
\label{alg:single_int}
\end{algorithm}

Proposition \ref{prop:min} reveals some characteristics of the persistent $1$-cycles computed by Algorithm \ref{alg:all_int} and \ref{alg:single_int}:

\begin{proposition}[Minimality Property]
\label{prop:min}
The persistent $1$-cycle $\sum_{g\in G}c_g$ computed by Algorithm \ref{alg:all_int} and \ref{alg:single_int} has the following property: There is no non-empty proper subset $G'$ of $G$ such that $\sum_{g\in G'}[c_g]=0$ in $\Hm_1(K_{d})$, where $d$ is the death index of the interval to which $\sum_{g\in G}c_g$ is associated.
\end{proposition}

Given that the minimal persistent $1$-cycles are not stable, it is not surprising that the cycles computed by Algorithm \ref{alg:all_int} are also not stable under perturbation.
Figure \ref{fig:instab_local} presents an example for which the perturbation of persistent $1$-cycles computed by Algorithm \ref{alg:all_int} cannot be properly bounded. 
We can construct a filtration by first forming the cycle $c$ then adding the other parts of the simplicial complex in Figure \ref{fig:instab_local}, making $v$ the last vertex and $e$ the last simplex. 
We then add a cone around $c_1$ to the filtration. 
Let the indices of $v$ and the apex vertex of the cone in the vertex sequence be $b$ and $d$. 
When $c$ is formed, the last edge $e'$ of $c$ is positive, and $c$ is chosen as the shortest cycle containing $e'$. 
When $e$ is added, we can make $c_1$ and $c_2$ be the two shortest cycles containing $e$. 
When $c_1$ is coned, if $c_1$ is chosen as the shortest cycle containing $e$, then the persistent $1$-cycle for the {interval $[b,d)$} would be $c_1$. 
Otherwise, the persistent $1$-cycle would be $c+c_2$. 
The length of $c$ can be arbitrary, so that the difference of the two persistent $1$-cycles can be arbitrary under the same insertion filtration of the same lower star filtration.

\begin{figure}
\centering
  \subfloat[]{\includegraphics[width=0.57\linewidth]{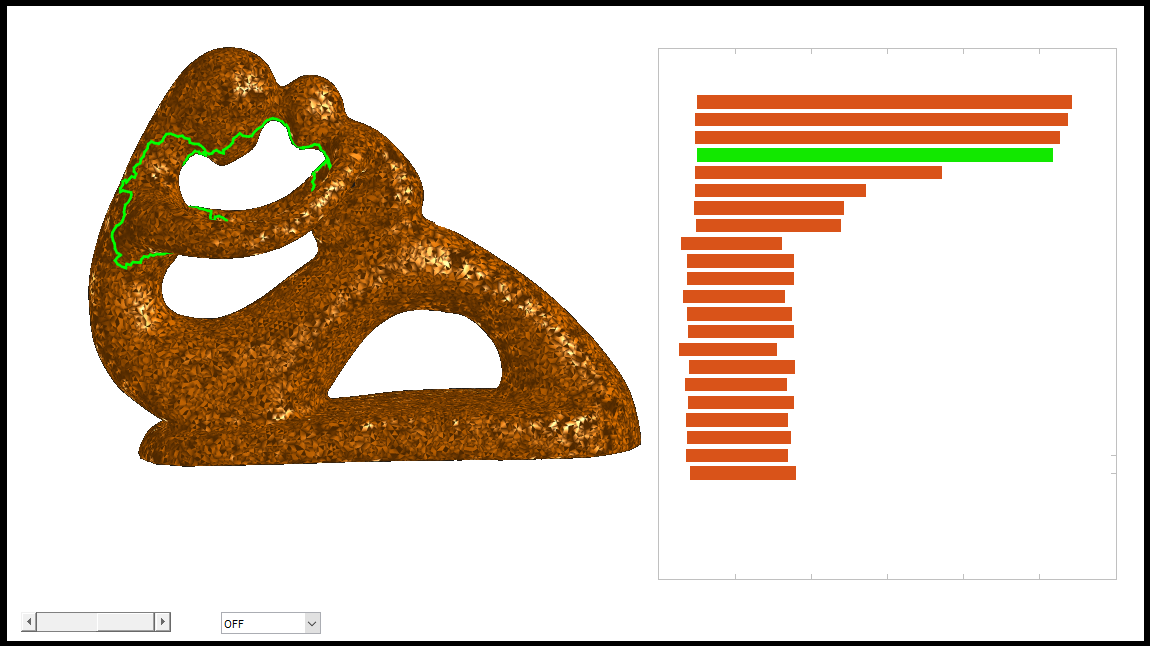}\label{fig:screenshot1}}
  \hspace{0.1cm}
  \subfloat[]{\includegraphics[width=0.37\linewidth]{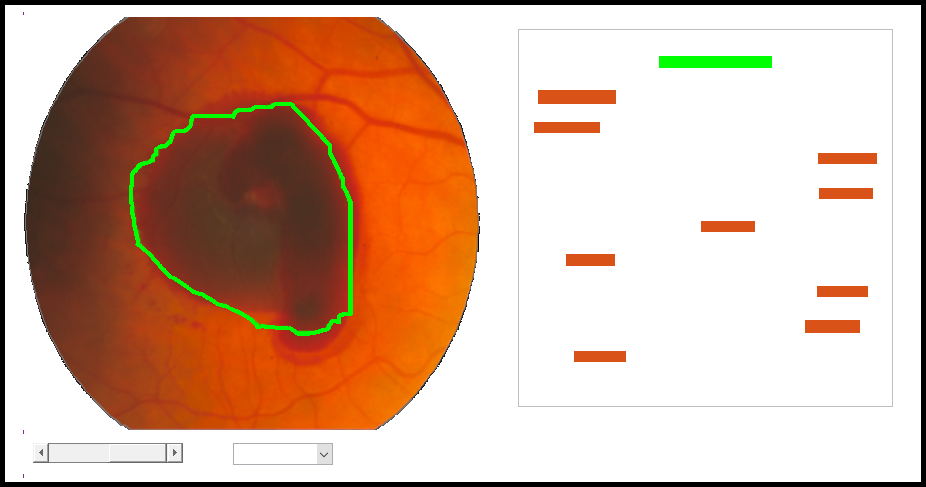}\label{fig:screenshot2}}
  \vspace{0.1cm}
   \caption{$\texttt{PersLoop}$ user interface demonstrating persistent $1$-cycles computed for a 3D point cloud (a) and a 2D image (b), where green cycles correspond to the chosen bars. }\label{fig:screenshot}
\end{figure}

\section{Results and experiments}
\label{sec:experiment}
Our software \texttt{PersLoop} implements Algorithm~\ref{alg:single_int}. 
Given a raw input, which is a 2D image or a 3D point cloud, and a filtration built from the raw input, the software first generates and plots the barcode of the filtration. 
The user can then click an individual bar to obtain the persistent $1$-cycle for the bar (Figure~\ref{fig:screenshot}).
The experiments on 3D point clouds and 2D images using the software show how our algorithm can find meaningful persistent $1$-cycles in several geometric and vision related applications.
\begin{figure*}[h]
\centering
\includegraphics[width=0.9\linewidth]{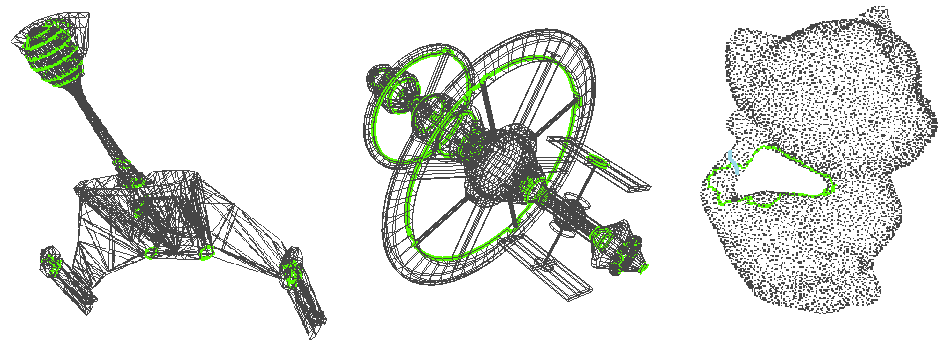}
  \caption{Persistent $1$-cycles (green) corresponding to long intervals computed for three different point clouds}
  \label{fig:3Dmodel}
\end{figure*}
\subsection{Persistent $1$-cycles for 3D point clouds}

We take a 3D point cloud as input and build a Rips filtration using the \texttt{Gudhi} library \cite{gudhi:urm}. 
As shown in Figure \ref{fig:3Dmodel}, persistent $1$-cycles computed for the three point clouds sampled from various models are tight and capture essential geometrical features of the corresponding persistent homology.
Note that our implementation of Algorithm \ref{alg:single_int} runs very fast in practice. For example, it took 0.3 secs to generate 50 cycles on a regular commodity laptop for the Botijo (Figure \ref{fig:intro1}) point cloud of size 10,000.

\subsection{Image segmentation and characterization using cubical complex}
In this section, we show the application of our algorithm on image segmentation and characterization problems. 
We interpret an image as a piecewise linear function on a 2-dimensional cubical complex. 
The cubical complex for an image has a vertex for each pixel, an edge connecting each pair of horizontally or vertically adjacent vertices, and squares to fill all the holes such that the complex becomes a disc.
The function values on the vertices are the density or color values of the corresponding pixels.
The lower star filtration \cite{edelsbrunner2010computational} of the PL function is then built and fed into our software.
We use the coning based annotation strategy~\cite{dey2014computing} to compute the persistence diagrams. 
In our implementation, a cubical tree, which is similar to the simplicial tree~\cite{Boissonnat2012}, is built to store the elementary cubes. 
Each elementary cube points to a row in the annotation matrix via the union find data structure. The simplicial counterpart of this association technique is described in~\cite{Boissonnat2013}.
\begin{figure}
\centering
  \subfloat[]{\includegraphics[width=0.35\linewidth]{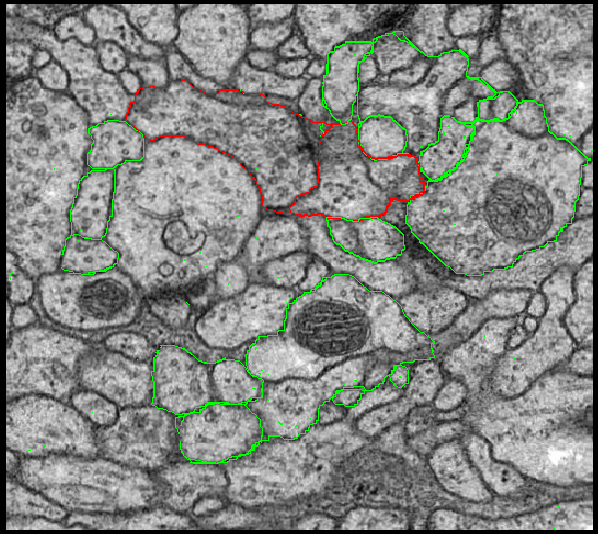}\label{fig:imageSeg1}}
  \hspace{0.1cm}
  \subfloat[]{\includegraphics[width=0.35\linewidth]{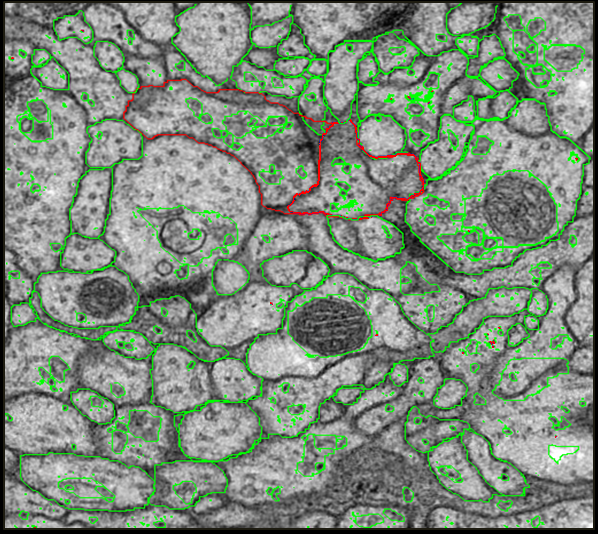}\label{fig:imageSeg2}}
  \vspace{0.1cm}
  \subfloat[]{\includegraphics[width=0.153\linewidth]{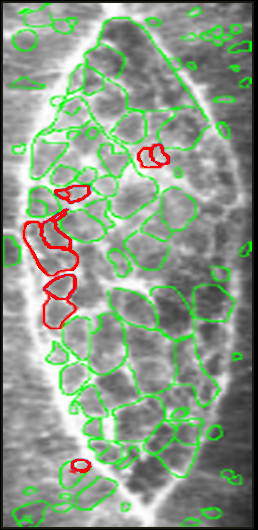}\label{fig:imageSeg3}}
   \caption{Persistent $1$-cycles computed for image segmentation. Green cycles indicate persistent $1$-cycles consisting of only one component ($|G|=1$) and red cycles indicate those consisting of multiple components ($|G|>1$). 
   (a,b) Persistent $1$-cycles for the top 20 and 350 longest intervals on the nerve dataset. 
   (c) Persistent $1$-cycles for the top 200 longest intervals on the Drosophila larvae dataset.}
   \label{fig:imageSeg}
\end{figure}

Our first experiment is the segmentation of a serial section Transmission Electron Microscopy (ssTEM) data set of the Drosophila first instar larva ventral nerve cord (VNC)~\cite{Cardona2010}. 
The segmentation result is shown in Figures~\ref{fig:imageSeg1} and~\ref{fig:imageSeg2}, from which we can see that the cycles are in exact correspondence to the membranes hence segment the nerve regions quite appropriately.
The difference between Figure ~\ref{fig:imageSeg1} and \ref{fig:imageSeg2} shows that longer intervals tend to have longer cycles.
Another similar application is the segmentation of the {low} magnification micrographs of a Drosophila  embryo~\cite{Kiehart2000}.
As seen in Figure~\ref{fig:imageSeg3}, the cycles corresponding to the top 200 longest intervals indicate that the larvae image is properly segmented.
 
We experiment on another dataset from the STARE project \cite{Hoover2003} to show how persistent $1$-cycles computed by our algorithm can be utilized for characterization of images. 
The dataset contains ophthalmologist annotated retinal images which are either healthy or suffering from diseases. 
Our aim is to automatically identify retinal and sub-retinal hemorrhages, which are black patches of blood accumulated inside eyes. 
Figures~\ref{fig:intro5} and~\ref{fig:screenshot2} show that a very tight cycle is derived around each dark hemorrhage blob even when the input is noisy.

\begin{figure}
\centering
  \subfloat[]{\includegraphics[width=0.15\linewidth]{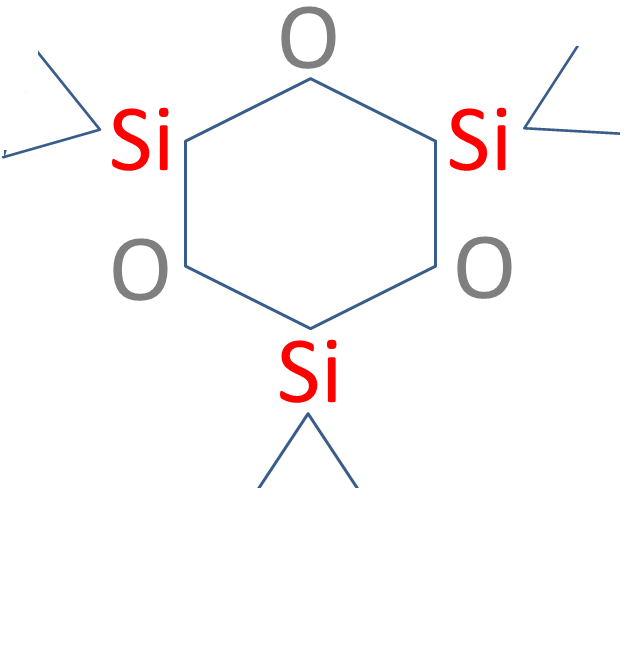}\label{fig:quartz1}}
  \hspace{0.1cm}
  \subfloat[]{\includegraphics[width=0.3\linewidth]{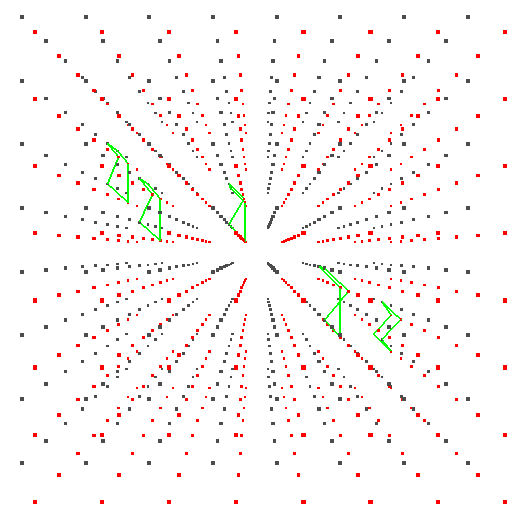}\label{fig:quartz2}}
  \vspace{0.1cm}
  \subfloat[]{\includegraphics[width=0.3\linewidth]{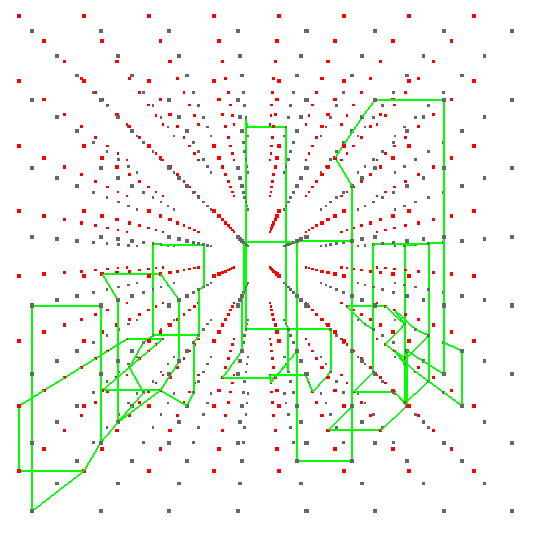}\label{fig:quartz3}}
    \vspace{0.1cm}
  \subfloat[]{\includegraphics[width=0.2\linewidth]{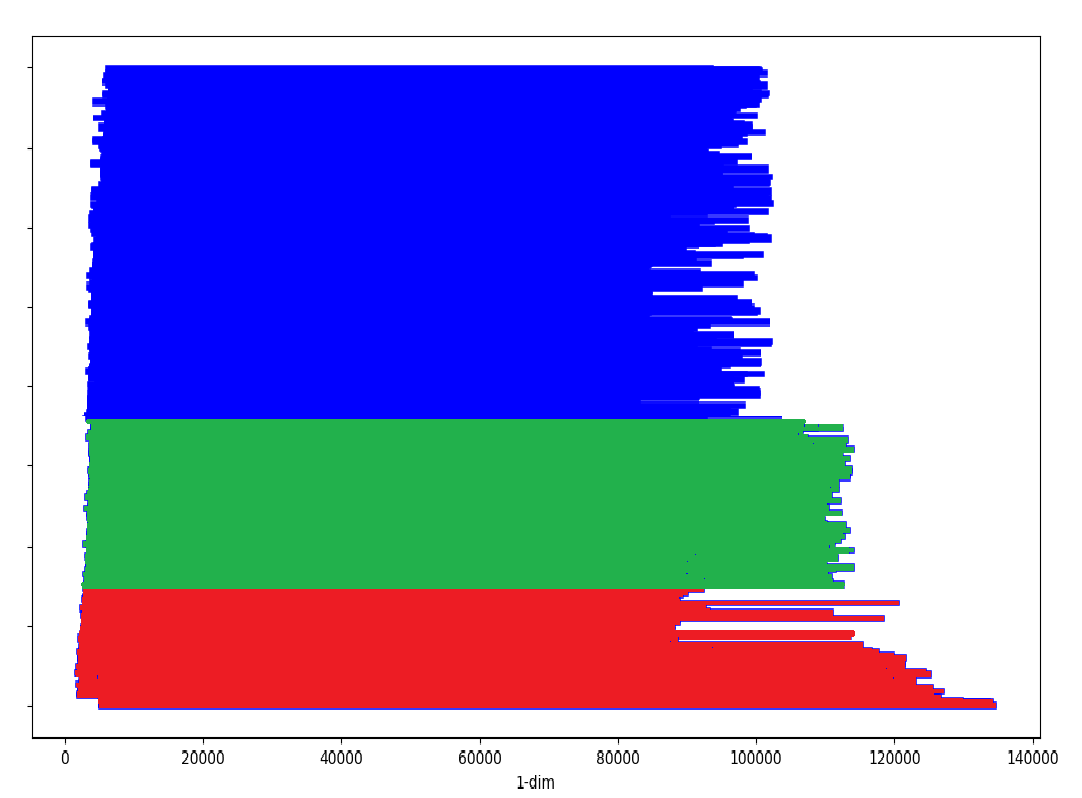}\label{fig:quartz4}}
   \caption{(a) Hexagonal cyclic structure of silicate glass. 
   (b) Persistent 1-cycles computed for the green bars with red points denoting silicate atoms and grey points denoting oxygen atoms. 
   (c) Persistent 1-cycles computed for the red bars.
   (d) Barcode for the filtration on silicate glass.}
   \label{fig:quartz}
\end{figure}

\subsection{Hexagonal structure of crystalline solids}
In this experiment, we use our persistent $1$-cycles to describe the crystalline structure of silicate glass (${SiO}_2$) commonly known as quartz. 
{Silicate glass} has a non-compact structure with three {silicon} and {oxygen} atoms arranged alternately in a hexagon as shown in Figure~\ref{fig:quartz1}.
We build a $8\times8\times8$ weighted point cloud with the silicon and oxygen atoms arranged according to the {space group} on the {crystal} structure as illustrated in Figure \ref{fig:quartz2}. 
The weights of the points correspond to {the atomic weights of the atoms}. 
On this weighted point cloud, we generate a filtration of weighted alpha complexes~\cite{Edelsbrunner:1992} by increasing $\alpha$ from $0$ to $\infty$. 

Persistent 1-cycles computed by our algorithm for this dataset reveal both the local and global structures of silicate glass. 
Figure~\ref{fig:quartz4} lists the barcode of the filtration we build and Figure~\ref{fig:quartz2} shows the persistent 1-cycles corresponding to the medium sized green bars in Figure~\ref{fig:quartz4}.
We can see on close observation that the cycles in Figure~\ref{fig:quartz2} are in exact accordance to the hexagonal cyclic structure of quartz shown in Figure~\ref{fig:quartz1}.
The larger persistent 1-cycles in Figure~\ref{fig:quartz3}, which span the larger lattice structure formed by our weighted point cloud, correspond to the longer red bars in Figure~\ref{fig:quartz4}. 
These cycles arise from the \textit{long-range order}{\footnote{Long-range order is the translational periodicity where the self-repeating structure extends infinitely in all directions}} of the crystalline solid.
This is evident from our experiment because if we increase the size of the input point cloud, these cycles grow larger to span the entire lattice.

\bibliographystyle{plain}
\bibliography{refs}

\begin{thebibliography}{10}

\bibitem{awodey2010category}
Steve Awodey.
\newblock {\em Category theory}.
\newblock Oxford University Press, 2010.

\bibitem{Boissonnat2013}
Jean{-}Daniel Boissonnat, Tamal~K. Dey, and Cl{\'{e}}ment Maria.
\newblock The compressed annotation matrix: an efficient data structure for
  computing persistent cohomology.
\newblock {\em CoRR}, abs/1304.6813, 2013.

\bibitem{Boissonnat2012}
Jean-Daniel Boissonnat and Cl\'{e}ment Maria.
\newblock The simplex tree: An efficient data structure for general simplicial
  complexes.
\newblock {\em 20th Annual European Symposium, Ljubljana,Slovenia},
  (2):731--742, 2012.

\bibitem{bubenik2014categorification}
Peter Bubenik and Jonathan~A Scott.
\newblock Categorification of persistent homology.
\newblock {\em Discrete \& Computational Geometry}, 51(3):600--627, 2014.

\bibitem{Cardona2010}
Saalfeld~S Cardona~A, Preibisch S, Schmid B, Cheng A, and Pulokas~J et~al.
\newblock An integrated micro- and macroarchitectural analysis of the
  drosophila brain by computer-assisted serial section electron microscopy.
\newblock {\em PLoS Biol}, 8, 2010.

\bibitem{Carlsson:2008}
Gunnar Carlsson, Tigran Ishkhanov, Vin Silva, and Afra Zomorodian.
\newblock On the local behavior of spaces of natural images.
\newblock {\em Int. J. Comput. Vision}, 76(1):1--12, January 2008.

\bibitem{chazal2016structure}
Fr{\'e}d{\'e}ric Chazal, Vin De~Silva, Marc Glisse, and Steve Oudot.
\newblock {\em The structure and stability of persistence modules}.
\newblock Springer, 2016.

\bibitem{chen2007quantifying}
Chao Chen and Daniel Freedman.
\newblock Quantifying homology classes ii: Localization and stability.
\newblock {\em arXiv preprint arXiv:0709.2512}, 2007.

\bibitem{chen2011hardness}
Chao Chen and Daniel Freedman.
\newblock Hardness results for homology localization.
\newblock {\em Discrete \& Computational Geometry}, 45(3):425--448, 2011.

\bibitem{cohen2005stability}
David Cohen-Steiner, Herbert Edelsbrunner, and John Harer.
\newblock Stability of persistence diagrams.
\newblock In {\em Proceedings of the twenty-first annual symposium on
  Computational geometry}, pages 263--271. ACM, 2005.

\bibitem{derksen2005quiver}
Harm Derksen and Jerzy Weyman.
\newblock Quiver representations.
\newblock {\em Notices of the AMS}, 52(2):200--206, 2005.

\bibitem{dey2014computing}
Tamal~K. Dey, Fengtao Fan, and Yusu Wang.
\newblock Computing topological persistence for simplicial maps.
\newblock In {\em Proceedings of the thirtieth annual symposium on
  Computational geometry}, page 345. ACM, 2014.

\bibitem{DHK11}
Tamal~K. Dey, Anil Hirani, and Bala Krishnamoorthy.
\newblock Optimal homologous cycles, total unimodularity, and linear
  programming.
\newblock {\em SIAM Journal on Computing}, 40(4):1026--1044, 2011.

\bibitem{Dey2017}
Tamal~K. Dey, Sayan Mandal, and William Varcho.
\newblock {Improved Image Classification using Topological Persistence}.
\newblock In Matthias Hullin, Reinhard Klein, Thomas Schultz, and Angela Yao,
  editors, {\em Vision, Modeling \& Visualization}. The Eurographics
  Association, 2017.

\bibitem{dey2010approximating}
Tamal~K. Dey, Jian Sun, and Yusu Wang.
\newblock Approximating loops in a shortest homology basis from point data.
\newblock In {\em Proceedings of the twenty-sixth annual symposium on
  Computational geometry}, pages 166--175. ACM, 2010.

\bibitem{Edelsbrunner:1992}
Herbert Edelsbrunner.
\newblock Weighted alpha shapes.
\newblock Technical report, Champaign, IL, USA, 1992.

\bibitem{edelsbrunner2010computational}
Herbert Edelsbrunner and John Harer.
\newblock {\em Computational topology: an introduction}.
\newblock American Mathematical Soc., 2010.

\bibitem{edelsbrunner2000topological}
Herbert Edelsbrunner, David Letscher, and Afra Zomorodian.
\newblock Topological persistence and simplification.
\newblock In {\em Foundations of Computer Science, 2000. Proceedings. 41st
  Annual Symposium on}, pages 454--463. IEEE, 2000.

\bibitem{emmett2016multiscale}
Kevin Emmett, Benjamin Schweinhart, and Raul Rabadan.
\newblock Multiscale topology of chromatin folding.
\newblock In {\em Proceedings of the 9th EAI international conference on
  bio-inspired information and communications technologies (formerly
  BIONETICS)}, pages 177--180. ICST (Institute for Computer Sciences,
  Social-Informatics and Telecommunications Engineering), 2016.

\bibitem{escolar2016optimal}
Emerson~G Escolar and Yasuaki Hiraoka.
\newblock Optimal cycles for persistent homology via linear programming.
\newblock In {\em Optimization in the Real World}, pages 79--96. Springer,
  2016.

\bibitem{Hoover2003}
A.~Hoover and M.~Goldbaum.
\newblock Locating the optic nerve in a retinal image using the fuzzy
  convergence of the blood vessels.
\newblock {\em IEEE Transactions on Medical Imaging}, 22(8):951--958, Aug 2003.

\bibitem{Kiehart2000}
Daniel~P. Kiehart, Catherine~G. Galbraith, Kevin~A. Edwards, Wayne~L. Rickoll,
  and Ruth~A. Montague.
\newblock Multiple forces contribute to cell sheet morphogenesis for dorsal
  closure in drosophila.
\newblock {\em The Journal of Cell Biology}, 149(2):471--490, 2000.

\bibitem{obayashi2017volume}
Ippei Obayashi.
\newblock Volume optimal cycle: Tightest representative cycle of a generator on
  persistent homology.
\newblock {\em arXiv preprint arXiv:1712.05103}, 2017.

\bibitem{papadimitriou1991optimization}
Christos~H Papadimitriou and Mihalis Yannakakis.
\newblock Optimization, approximation, and complexity classes.
\newblock {\em Journal of computer and system sciences}, 43(3):425--440, 1991.

\bibitem{gudhi:urm}
{The GUDHI Project}.
\newblock {\em {GUDHI} User and Reference Manual}.
\newblock {GUDHI Editorial Board}, 2015.

\bibitem{wu2017optimal}
Pengxiang Wu, Chao Chen, Yusu Wang, Shaoting Zhang, Changhe Yuan, Zhen Qian,
  Dimitris Metaxas, and Leon Axel.
\newblock Optimal topological cycles and their application in cardiac
  trabeculae restoration.
\newblock In {\em International Conference on Information Processing in Medical
  Imaging}, pages 80--92. Springer, 2017.

\end{thebibliography}


\appendix

\section{Proof of Theorem \ref{thm:nphard}}
\label{app:nphard-pf}
Similar to \cite{chen2007quantifying}, we will reduce the NP-hard MAX-2SAT \cite{papadimitriou1991optimization} problem to LST-PERS-CYC. The MAX-2SAT is defined as:

\begin{problem}[MAX-2SAT] 
Given $N$ variables $x_1,\ldots,x_N$ and $M$ clauses $c_1,\ldots,c_M$, with the clauses being the disjunction of at most two variables. Find an assignment of Boolean values to all the variables such that the maximal number of clauses are satisfied.
\end{problem}

\begin{proof}[Proof of Theorem \ref{thm:nphard}]
We will reduce MAX-2SAT to LST-PERS-CYC. Given an instance of MAX-2SAT, 
we first construct a simplicial complex $K$ as in \cite{chen2007quantifying}, 
by forming a triangulated cylinder $C_i$ for each variable $x_i$ and a cycle $w_j$ for each clause $c_j$, 
such that the two ends $z_i$ and $z_i'$ of $C_i$ correspond to $x_i$ and $\neg x_i$ and are the only two cycles with the least number of edges of their homology class in $C_i$. 
To make the process clearer, our construction of the cycles $z_i$, $z_i'$, and $w_j$ are a little different from \cite{chen2007quantifying}.
Each $z_i$ or $z_i'$ has $3M$ edges and $M$ of them are bijectively assigned to the $M$ clauses, 
such that in between each two consecutive edges assigned to some clauses, 
there are two edges which are not assigned to any clause. 
For a clause cycle $w_j=(z_i',z_k)$ (do the similar for other cases), 
we assign three edges to $w_j$ and pick one edge to be shared with the edges in $z_i'$ and $z_k$ assigned to $w_j$. 
Let $\overline{z}=\sum^N_{i=1}z_i + \sum^M_{j=1}w_j$, 
then our construction will make it true that, there is an assignment of Boolean values making $k$ clauses satisfied if and only if there is a cycle in $[\overline{z}]$ with $3MN+3M-2k$ edges.

Next we are going to construct a filtration $\Fcal'$ of a complex $K'$, where $K\subseteq K'$. We first construct a filtration $\Fcal$ of $K$, with the only restriction: Pick an edge $e$ of a clause cycle, which is not shared with any end cycle of the variable cylinders, and take $e$ as the last simplex added to the filtration. 
To construct $\Fcal'$ and $K'$, we need to find all simple cycles of $\overline{z}$. A simple cycle is defined as a cycle such that, each vertex has degree $2$ and there is only one connected component in the cycle. We can use a DFS-based algorithm to find all simple cycles for $\overline{z}$: Treat $\overline{z}$ as graph and run DFS on the graph. Find a non-DFS-tree edge $(v_1,v_2)$ of $\overline{z}$, then find the lowest common ancestor $w$ of $v_1$ and $v_2$ in the DFS tree. The paths in the DFS tree from $w$ to $v_1$ and $w$ to $v_2$, plus the edge $(v_1,v_2)$, form a simple cycle of $\overline{z}$. Delete the simple cycle from the graph and repeat the above process until the graph becomes empty. 

For each simple cycle $\overline{c}$ of $\overline{z}$, 
we attach a cylinder $\overline{C}$ to $\overline{c}$ such that, one end of $\overline{C}$ is $\overline{c}$, 
the other end of $\overline{C}$ is a quadrilateral, 
and all the other edges of $\overline{C}$ connect $\overline{c}$ to the quadrilateral. 
An example of such a cylinder connecting a dodecagon and a quadrilateral is illustrated in Figure \ref{fig:cylinder}. 
After all the cylinders are attached to the simple cycles, we get a simplicial complex $K_1$. 
We can append the simplices of $K_1\smallsetminus K$ to $\Fcal'$, 
to get a filtration $\Fcal_1$ of $K_1$. 
Since $K_1$ deformation retracts onto $K$, 
all negative triangles of $K_1\smallsetminus K$ are paired with an edge of $K_1\smallsetminus K$. 
We then construct a simplicial complex whose boundary is the sum of all the quadrilaterals and an outer cycle $c'$, as in Figure \ref{fig:quad_outer_complex}, 
and attach this simplicial complex to $K_1$ by gluing the quadrilaterals, to get a simplicial complex $K_2$. To form a filtration $\Fcal_2$ of $K_2$, we first append the red edges in Figure \ref{fig:quad_outer_complex} to $\Fcal_1$, then append all the other simplices of $K_2\smallsetminus K_1$. Finally, we form a cone around $c'$ to get $K'$ and append the missing simplices to get the filtration $\Fcal'$.

\begin{figure}
	\centering
	\subfloat
      []{\includegraphics[width=0.15\linewidth]{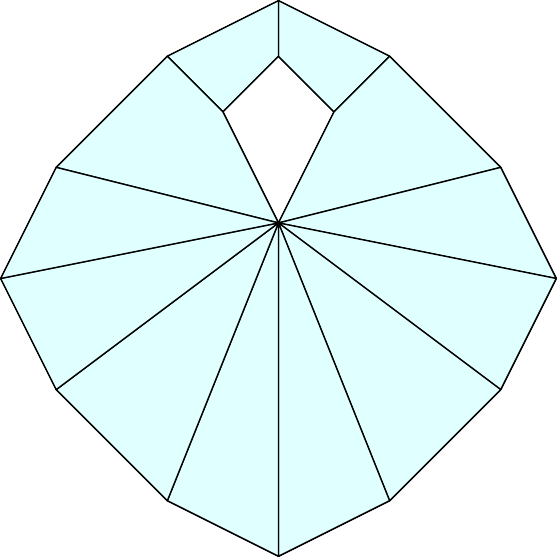}\label{fig:cylinder}} \hspace{3em}
    \subfloat
      []{\includegraphics[width=0.3\linewidth]{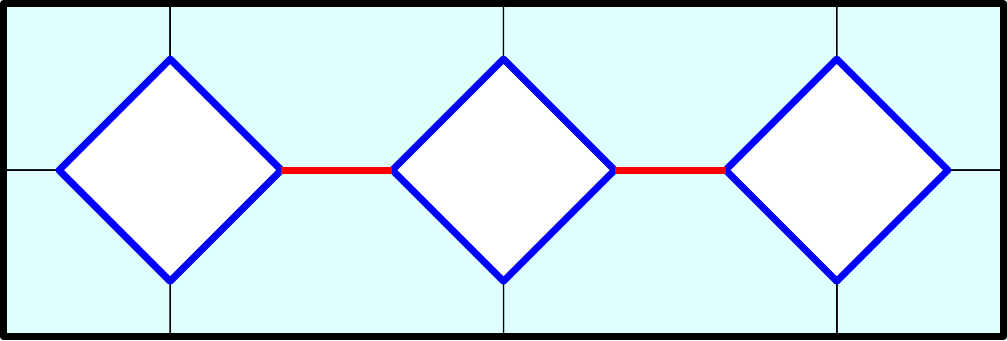}\label{fig:quad_outer_complex}}
	\caption{(a) A cylinder connecting a dodecagon and a quadrilateral. (b) A simplicial complex whose boundary is the sum of three quadrilaterals (blue) and an outer cycle (bold). Some polygons in the figure are not triangulated.}
\end{figure}

Let $t$ be the last triangle in $\Fcal'$, then it is true that $K'\smallsetminus t$ deformation retracts to the union of $K_1$ and the red edges. This indicates that all negative triangles of $K'\smallsetminus K_1$, other than $t$, are paired with edges of $K'\smallsetminus K_1$. Let the index of $e$ in $\Fcal'$ be $b$ and the index of $t$ in $\Fcal'$ be $d$, we claim that $[b,d)$ is an interval of $\persdgm_1(\Fcal')$. To prove this, first note that $\overline{z}$ is born in $K_b$ and becomes a boundary in $K_{d}$. By the time $t$ is added, $e$ is unpaired. So by Algorithm \ref{alg:diag}, $t$ must be paired with $e$.

Now we have constructed an instance of LST-PERS-CYC, from an instance of MAX-2SAT: 
Given the filtration $\Fcal'$ and the interval $[b,d)\in\persdgm_1(\Fcal')$, 
find a persistent $1$-cycle with the least number of edges. 
We then prove that the answer to LST-PERS-CYC is also the answer to MAX-2SAT. First note that the map $\Hm_1(K_b)\rightarrow \Hm_1(K_{d-1})$ is injective. 
This means that any persistent $1$-cycle for $[b,d)$ must be homologous to $\overline{z}$ in $K_{b}$, as they are homologous in $K_{d-1}$. 
It follows that computing the minimal persistent $1$-cycle of $[b,d)$ is equivalent to computing the minimal cycle of the homology class $[\overline{z}]$ in $K_b$, 
which is in turn equivalent to computing the answer for the original MAX-2SAT problem. 
Then we have had a reduction from MAX-2SAT to LST-PERS-CYC. 
Furthermore, the reduction is in polynomial time and the size of the constructed instance of LST-PERS-CYC is a polynomial function of that of MAX-2SAT, so LST-PERS-CYC is NP-hard.
\end{proof}


\end{document}